\documentclass[a4paper, 10pt, twocolumn]{article}

\usepackage{authblk}
\usepackage{graphicx}
\usepackage{amsmath}
\usepackage{amsthm}
\usepackage{amssymb}
\usepackage{hyperref}

\graphicspath{{../}}

\newtheorem{proposition}{Proposition}

\title{Defensive Resource Allocation in Social Networks}

\author[1]{Antonia Maria Masucci
\thanks{Email: \href{mailto:antonia.masucci@inria.fr}{antonia.masucci@inria.fr}}}
\affil[1]{INRIA Paris-Rocquencourt\\ Domaine de Voluceau B.P. 105\\ 78153 Le Chesnay\\ France}
\author[2]{Alonso Silva
\thanks{Email: \href{mailto:alonso.silva@alcatel-lucent.com}{alonso.silva@alcatel-lucent.com}
To whom correspondence should be addressed.}}
\affil[2]{Alcatel-Lucent Bell Labs France\\ Centre de Villarceaux\\ Route de Villejust\\ 91620 Nozay\\ France}
\date{}

\date{}

\begin{document}

\maketitle

\begin{abstract}
In this work,
we are interested on the analysis of
competing marketing campaigns
between an incumbent who dominates
the market and a challenger
who wants to enter the market.
We are interested in (a) the simultaneous
decision of how many resources to allocate
to their potential customers to advertise their products
for both marketing campaigns,
and (b) the optimal allocation
on the situation in which the incumbent
knows the entrance of the challenger and thus
can predict its response.
Applying results from game theory,
we characterize these optimal strategic resource allocations
for the voter model of social networks.
\end{abstract}

\section{Introduction}

In contrast to mass marketing, which promotes a product
indiscriminately to all potential customers,
direct marketing promotes a product only
to customers likely to be profitable.
Focusing on the latter, Domingos and Richardson~\cite{DomingosR2001, RichardsonD2002}
incorporated the influence of peers on the
decision making process of potential customers
deciding between different products or services promoted by
competing marketing campaigns.
This aggregated value of a customer has been called the {\sl network value of a customer}.

If we consider that each customer
makes a buying decision independently of
every other customer, we should only consider
his intrinsic value (i.e., the expected profit from sales to him).
However, if we consider the often strong influence of her friends,
acquaintances, etc., then we should
incorporate this peer influence to his value for the marketing campaigns.

In the present work,
our focus is different from
previous works where
their interest is to which potential
customers to market,
while in our work
is on how many resources to allocate to market
to potential customers.
Moreover, we are interested
on the scenario when two competing
marketing campaigns
need to 
decide how many resources to allocate
to potential customers to advertise
their products
either simultaneously or
where the incumbent can foresee
the arrival of the challenger and
commit to a strategy.
The process and dynamics by which
influence is spread is given
by the voter model.

\subsection{Related Works}

The (meta) problem of influence maximization
was first defined by Domingos and Richardson~\cite{DomingosR2001,RichardsonD2002},
where they studied a probabilistic setting of this problem and
provided heuristics to compute a spread maximizing set.
Based on the results of Nemhauser et al.~\cite{NemhauserWF1978},
Kempe et al.~\cite{Kempe2003,Kempe2005}
and Mossel and Roch~\cite{MosselR2007}
proved that for very natural activation functions,
the function of the expected number of active nodes
at termination is a submodular function
and thus can be approximated through
a greedy approach with a $(1-1/e-\varepsilon)$-approximation
algorithm for the spread maximization set problem.
A slightly different model but with similar flavor,
the voter model, was introduced by Clifford and Sudbury~\cite{CliffordS1973}
and Holley and Liggett~\cite{HolleyL1975}.
In that model of social network, Even-Dar and Shapira~\cite{EvenDar2007}
found an exact solution to the spread
maximization set problem when all the nodes
have the same cost.

Competitive influence in social networks has been studied
in other scenarios.
%
Bharathi et al.~\cite{BharathiKS2007}
proposed a generalization of the independent
cascade model~\cite{GoldenbergLM2001}
and gave a $(1-1/e)$ approximation
algorithm for computing the best response
to an already known opponent's strategy.
Sanjeev and Kearns~\cite{GoyalK2012}
studied the case of two players
simultaneously choosing
some nodes to initially seed
while considering two independent functions
for the consumers
denoted switching function and
selection function.
Borodin et al.~\cite{BorodinFO2010}
showed that for a broad family of competitive
influence models it is NP-hard
to achieve an approximation that
is better than the square root
of the optimal solution.
Chasparis and Shamma~\cite{ChasparisS2010}
found optimal advertising policies using dynamic programming
on some particular models of social networks.

Within the general context of competitive contests,
there is an extensive literature (see e.g.~\cite{GrossW1950,Roberson2006,masucciS2014,SchwartzLS2014}),
however their focus is mainly on the case when the contest success function
is given by the marketing campaign that put the maximum resources.
In that case, Powell~\cite{Powell2008} studied the sequential,
non-zero sum game who has a pure strategy subgame perfect equilibrium where
the defender always plays the same pure strategy in any equilibrium,
and the attacker's equilibrium response is generically unique
and entails no mixing.
Friedman~\cite{Friedman1958} studied the Nash equilibrium
when the valuations for both marketing campaigns are the same.

\section{Model}

Consider two firms: an incumbent (or defender)~$D$ and a challenger (or attacker)~$A$.
Consider the set $\mathcal{N}=\{1,2,\ldots,n\}$ of potential customers.
The challenger decides to launch a viral marketing campaign at time $t_0=0$
(we will also refer to the viral marketing campaign as an attack).
The budget for the challenger is given by $B_A\ge0$.
The incumbent~$D$ decides to allocate a budget $B_D\ge0$ at time $t_0$ to prevent its
customers to switch.
The players of the game are the competing marketing campaigns and the nodes correspond
to the potential customers.

The strategy for player~$i$, where $i\in\{D,A\}$, consists on an allocation vector ${\bf x}_i=(x_{i,1},x_{i,2}\ldots,x_{i,n})$
where $x_{i,j}$ represents the budget allocated by player $i$ to customer~$j$
(e.g., through promotions or offers).
Therefore, the set of strategies
is given by the $(n-1)$-dimensional simplex
\begin{align*}
\Delta_i=\{&(x_{i,1},\ldots,x_{i,n}):\\
&x_{i,j}\ge0,\forall\,1\le j\le n\quad\textrm{and}\quad\sum_{j=1}^nx_{i,j}=B_i\}.
\end{align*}

We consider each potential customer as a component contest.
Let $p_{i,j}(x_{i,j},x_{-i,j})$, henceforth the
contest success function (CSF), denote the
probability that player $i$ wins component
contest $j$ when player $i$ allocates $x_{i,j}$
resources and the adversary player $-i$ allocates $x_{-i,j}$
resources to component contest~$j$. We assume
that the CSF for a player~$i$ is proportional to the 
share of total advertising expenditure on customer $j$, i.e.,
\begin{equation}\label{eq:fra1}
p_{i,j}(x_{i,j},x_{-i,j})=
\left\{
\begin{array}{cl}
\frac{x_{i,j}}{x_{i,j}+x_{-i,j}} & \textrm{if }x_{i,j}\neq 0\textrm{ or } x_{i,j}\neq0,\\
\frac12 & \textrm{if }x_{i,j}=x_{-i,j}=0.
\end{array}
\right.
\end{equation}

Both firms may have different valuations for different customers.
The intrinsic value for player $i$ of customer $j$ is given by $w_{i,j}$
where $i\in\{D,A\}$ and $1\le j\le N$.
The intrinsic payoff function for player $i$ is given by
\begin{equation}
\pi({\bf x}_i,{\bf x}_{-i},{\bf w}_i)=\sum_{j=1}^n w_{i,j}p_{i,j}(x_{i,j},x_{-i,j}),
\end{equation}
where ${\bf w}_i=(w_{i,1},w_{i,2},\ldots,w_{i,n})$.

We are interested as well on the network value of a customer.
In the next subsection we will compute this.

\subsection{Network value of a customer}

Let $G=(V,E)$ be an undirected graph with self-loops
where $V$ is the set of nodes in the graph
which represent the potential customers of
the competing marketing campaigns
and $E$ is the set of edges which
represent the influence between individuals.
We denote by $\lvert S\rvert$ the cardinality of the set~$S$,
by the index $i$ to one of the two players ($D$ or $A$)
and by the index $-i$ to the opponent of player $i$.
We consider that the graph $G$ has $n$ nodes, i.e. $\lvert V\rvert=n$.
For a node $j\in V$, we denote by $N(j)$ the set of neighbors
of $j$ in~$G$, i.e. $N(j)=\{j'\in V: \{j,j'\}\in E\}$
and by $d_j$ the degree of node $j$, i.e. $d_j=\lvert N(j)\rvert$.

We consider two labeling functions for a node $j\in V$
given by its initial preference
between different players, $D$ or $A$,
denoted by functions $f_A^0$ and $f_D^0$ respectively.
We denote by $f^0_i(v)=1$ when node $v\in V$ prefers
the product promoted by marketing campaign $i$.
We consider that every customer has an initial preference
between the firms, i.e.~$f^0_i=1-f^0_{-i}$.

We assume that the initial preference for a customer~$j$ is proportional to the 
share of total advertising expenditure on customer~$j$, i.e.,
\begin{equation}
f_i^0(j)=
\left\{
\begin{array}{rl}
1 & \textrm{with probability } p_{i,j}(x_i,y_i)\\
0 & \textrm{with probability } p_{-i,j}(x_i,y_i)\\
\end{array}
\right.
\end{equation}
where the function $p_{i,j}(\cdot,\cdot)$ is given by eq.~\eqref{eq:fra1}.

The evolution of the system will be described by the voter model.
Starting from any arbitrary initial preference assignment to the vertices
of $G$, at each time $t\ge 1$, each node picks uniformly at random
one of its neighbors and adopts its opinion.
In other words, starting from any assignment \mbox{$f_i^0: V\to\{0,1\}$},
we inductively define
\begin{equation}
f_i^{t+1}(j)=
\left\{
\begin{array}{rl}
1 & \textrm{with prob. }\frac{\lvert\{j'\in N(j): f_i^t(j')=1\}\rvert}{\lvert N(j)\rvert},\\
0 & \textrm{with prob. }\frac{\lvert\{j'\in N(j): f_i^t(j')=0\}\rvert}{\lvert N(j)\rvert}.\\
\end{array}
\right.
\end{equation}

For player $i$ and target time $\tau$, the expected payoff is given~by
\begin{equation}
\mathbf{E}
\left[\sum_{j\in V} w_{i,j}f_i^\tau(j)\right].
\end{equation}

We notice that in the voter model, the probability
that node $j$ adopts the opinion of one its neighbors $j'$
is precisely $1/\lvert N(j)\rvert$. Equivalently, this is the probability
that a random walk of length $1$ that starts at $j$
ends up in $j'$.
Generalizing this observation by induction on $t$,
we obtain the following proposition.

\begin{proposition}[Even-Dar and Shapira~\cite{EvenDar2007}]
Let $p_{j,j'}^t$ denote the probability that a random
walk of length $t$ starting at node $j$
stops at node $j'$.
Then the probability that after $t$ iterations of the voter model,
node $j$ will adopt the opinion that node $j'$ had at time $t=0$
is precisely $p_{j,j'}^t$.
\end{proposition}

Let $M$ be the normalized transition matrix of $G$, i.e.,
$M(j,j')=1/\lvert N(j)\rvert$ if $j'\in N(j)$.
By linearity of expectation, we have that for player $i$
\begin{equation}
\mathbf{E}\left[
\sum_{j\in V} w_{i,j}f_i^\tau(j)\right]=\sum_{j\in V}w_{i,j}\mathbf{P}[f_i^\tau(j)=1].
\end{equation}

The probability that a random walk
of length~$t$ starting at~$j$ ends in $j'$,
is given by the $(j,j')$-entry of the matrix~$M^t$.
Then
\begin{align*}
\mathbf{P}[f_i^t(j)=1]&=\sum_{j'\in V} p^t_{j,j'}\mathbf{P}[f_i^0(j')=1]\\
&=\sum_{j'\in V} M^t(j,j')\mathbf{P}[f_i^0(j')=1],
\end{align*}
and therefore,
\begin{equation}\label{eq:bamako}
\mathbf{E}\left[
\sum_{j\in V} w_{i,j} f_i^t(j)
\right]=
\sum_{j\in V}\sum_{j'\in V} w_{i,j}M^t(j,j')\mathbf{P}[f_i^0(j')=1].
\end{equation}
We know that
$\mathbf{P}[f_i^0(j')=1]=p_{i,j'}(x_{i,j'},x_{-i,j'})$.
Therefore, eq.~\eqref{eq:bamako} becomes
\begin{equation}
\sum_{j\in V}\sum_{j'\in V} w_{i,j}M^t(j,j')p_{i,j'}(x_{i,j'},x_{-i,j'}).
\end{equation}

Therefore, the expected payoff for player $i$ is given by
\begin{equation}\label{eq:gabon}
F_i({\bf x}_i,{\bf x}_{-i},{\bf v}_i)=\sum_{j'=1}^n v_{i,j'}\frac{x_{i,j'}}{x_{i,j'}+x_{-i,j'}},
\end{equation}
where ${\bf v}_i=(v_{i,1},v_{i,2},\ldots,v_{i,n})$.
\begin{equation}
v_{i,j'}=\sum_{j\in V} w_{i,j}M^t(j,j'),
\end{equation}
corresponds to the network value of customer~$j'$ at time $t$.
The previous expression is
subject to the constraint~${\bf x}_i\in\Delta_i$.

\section{Results}

From the previous section,
we are able to compute the
network value of each customer,
therefore we can restrict ourselves
to work with these values.
With the next proposition,
we are able to determine the best response function
for player $i$ considering the network value of each customer
at a target time $\tau$ given that the strategy of the
opponent $-i$ is~${\bf x}_{-i}$.

\begin{proposition}[Friedman~\cite{Friedman1958}]\label{prop:bamako}
The best response function for player $i$, given that
player $-i$ strategy is ${\bf x}_{-i}$, is:
\begin{equation}
x^*_{i,j}=-x_{-i,j}+(B_i+B_{-i})\frac{\sqrt{v_{i,j}x_{-i,j}}}{\sum_{k=1}^n\sqrt{v_{i,k}x_{-i,k}}}.
\end{equation}
\end{proposition}

From the previous proposition, we obtain that the best response functions for players
$A$ and $D$ are given by
\begin{eqnarray}
x^*_{A,j}&=-x_{D,j}+(B_D+B_A)\frac{\sqrt{v_{A,j}x_{D,j}}}{\sum_{k=1}^n\sqrt{v_{A,k}x_{D,k}}},\label{eq:boh1}\\
x^*_{D,j}&=-x_{A,j}+(B_A+B_D)\frac{\sqrt{v_{D,j}x_{A,j}}}{\sum_{k=1}^n\sqrt{v_{D,k}x_{A,k}}}.\label{eq:boh2}
\end{eqnarray}

In the next proposition, we assume that the valuations of one of the players
are proportional (bigger ($\alpha>1$), smaller ($\alpha<1$) or equal ($\alpha=1$))
to the valuations of the adversary player.
In that case, we have the following proposition.
\begin{proposition}\label{prop:sierra}
If $v_{i,j}=\alpha v_{-i,j}$ $\forall j\in\{1,\ldots,n\}$, $i\in\{A,D\}$,
with $\alpha>0$, then the Nash equilibrium of the game is given by
\begin{equation}
x_{i,j}=B_i\frac{v_{i,j}}{V_i},
\end{equation}
where $V_i=\sum_{j=1}^n v_{i,j}$.
\end{proposition}

\begin{proof}
The proof for the case $\alpha=1$ is given by Friedman~\cite{Friedman1958}.
For that case, from eq.~\eqref{eq:gabon},
the objective function for player $i$ is given by
$F_i({\bf x}_i,{\bf x}_{-i},{\bf v}_{-i})$. 
For $\alpha>0$, we have that the objective function,
\begin{eqnarray}
F_i({\bf x}_i,{\bf x}_{-i},{\bf v}_i)&=\sum_{j=1}^n v_{i,j}\frac{x_{i,j}}{x_{i,j}+x_{-i,j}},\\
&=\alpha \sum_{j=1}^n v_{-i,j}\frac{x_{i,j}}{x_{i,j}+x_{-i,j}}.\\
&=\alpha F_i({\bf x}_i,{\bf x}_{-i},{\bf v}_i),
\end{eqnarray}
and thus we conclude from the previous case.
\end{proof}

The previous result not only gives explicitly the Nash
equilibrium under some constraints, 
but it proves that a scaling factor for every contest
does not change the Nash equilibrium strategies of the players.
 
An interesting property, that we will exploit in
the following is that if for one of the players
all the contests have the same valuation,
then for every two equal valuation contests for
the adversary player, the equilibrium allocation for each player
in the two contests are equal.

\begin{proposition}\label{prop:uganda}
Assume that $v_{A,\ell}=v\quad\forall 1\le \ell\le n$.
If there exist $k, k'\in\{1,\ldots,n\}$ such that $v_{D,k}=v_{D,k'}$,
then \mbox{$x_{A,k}=x_{A,k'}$} and $x_{D,k}=x_{D,k'}$.
\end{proposition}

\begin{proof}
For $j\in\{1,\ldots,n\}$,
since $v_{A,\ell}=v$ \mbox{$\forall 1\le \ell\le n$}, from eq.~\eqref{eq:boh1} we have
\begin{equation}\label{eq:boh3}
x_{A,j}+x_{D,j}=\gamma\sqrt{x_{D,j}},
\end{equation}
where
\begin{equation}
\gamma=\frac{(B_A+B_D)}{\sum_{k=1}^n\sqrt{x_{D,k}}}
\end{equation}
does not depend on $j$.

From the difference between eq.~\eqref{eq:boh1} and eq.~\eqref{eq:boh2},
\begin{equation}\label{eq:boh4}
x_{D,j}=\kappa^2 v_{D,j}x_{A,j},
\end{equation}
where
\begin{equation}
\kappa=\frac{\sum_{k=1}^n\sqrt{x_{D,k}}}{\sum_{k=1}^n\sqrt{v_{D,k}x_{A,k}}}
\end{equation}
does not depend on $j$.

Replacing eq.~\eqref{eq:boh4} on eq.~\eqref{eq:boh3},
\begin{equation}
(1+\kappa^2v_{D,j})x_{A,j}=\gamma\kappa\sqrt{v_{D,j}x_{A,j}},
\end{equation}
or equivalently,
\begin{equation}\label{eq:boh5}
x_{A,j}=\frac{\gamma^2\kappa^2v_{D,j}}{(1+\kappa^2v_{D,j})^2}.
\end{equation}
If there exist $k, k'\in\{1,\ldots,n\}$ such that $v_{D,k}=v_{D,k'}\equiv w$,
then from eq.~\eqref{eq:boh5},
\begin{equation}
x_{A,k}=\frac{\gamma^2\kappa^2w}{(1+\kappa^2w)^2}=x_{A,k'}.
\end{equation}
From eq.~\eqref{eq:boh4}, we also obtain
\begin{equation}
x_{D,k}=\kappa^2 \frac{\gamma^2\kappa^2w^2}{(1+\kappa^2w)^2}=x_{D,k'}.
\end{equation}
\end{proof}

In Proposition~\ref{prop:sierra}, we proved that
a scaling factor between the valuations of the players
does not change their equilibrium strategies.
However, we will see in the following proposition
that this situation is unusual.
Actually, even in the case of two communities within
the social network, where the valuations for the attacker are the same
and the valuations for the defender are different for each community,
the players have very different strategies than the previously considered.

\begin{proposition}\label{prop:zaire1}
Assume $n$ is even, so there exists \mbox{$m\in\mathbf{N}\setminus\{0\}$} such that $n=2m$.
Assume that $v_{A,\ell}=v\quad\forall 1\le\ell\le n$,
$v_{D,k}=\alpha v\quad\forall 1\le k\le m$ and $v_{D,k'}=\beta v\quad\forall m+1\le k'\le n$,
for $\alpha,\beta>0$.
Then the Nash equilibrium is given by
\begin{align*}
x_{A,1}&=\ldots=x_{A,m}=x_A^*,\\
x_{A,m+1}&=\ldots=x_{A,n}=(B_A/m-x_A^*),\\
x_{D,1}&=\ldots=x_{D,m}=\frac{B_D}{m}\frac{\alpha x_A^*}{\alpha x_A^*+\beta (B_D/m-x_A^*)},\\
x_{D,m+1}&=\ldots=x_{D,n}=\frac{B_D}{m}\frac{\beta (B_D/m-x_A^*)}{\alpha x_A^*+\beta (B_D/m-x_A^*)}.
\end{align*}
where $x_A^*$ is unique (its expression is not important and thus it is given in the Appendix).
%
%
\end{proposition}

\begin{proof}
From Proposition~\ref{prop:uganda}, we have that
\begin{align*}
&x_{A,1}=x_{A,2}=\ldots=x_{A,m}\equiv x,\\
&x_{A,m+1}=x_{A,m+2}=\ldots=x_{A,n}\equiv y.
\end{align*}

From eq.~\eqref{eq:boh4},
\begin{equation}
\frac{x_{D,j}}{v_{D,j}x_{A,j}}=\kappa^2,
\end{equation}
where $\kappa$ does not depend on $j$.
Therefore for all $k,k'\in\{1,\ldots,n\}$,
\begin{equation}
\frac{x_{D,k}}{v_{D,k}x_{A,k}}=
\frac{x_{D,k'}}{v_{D,k'}x_{A,k'}}.
\end{equation}
Let us consider $k\in\{1,\ldots,m\}$ and $k'\in\{m+1,\ldots,n\}$,
then from the previous equation
\begin{equation}\label{eq:leo1}
x_{D,k}=\frac{\alpha x}{\beta y} x_{D,k'}.
\end{equation}
We know that
$\sum_{j=1}^n x_{D,j}=B_D$,
thus
\begin{equation}
m\left(\frac{\alpha x}{\beta y}+1\right)x_{D,k'}=B_D.
\end{equation}
Then
\begin{equation}\label{eq:leo2}
x_{D,m+1}=x_{D,m+2}=\ldots=x_{D,n}=\frac{B_D}{m}\frac{\beta y}{\alpha x+\beta y},
\end{equation}
and from eq.~\eqref{eq:leo1},
\begin{equation}\label{eq:leo3}
x_{D,1}=x_{D,2}=\ldots=x_{D,m}=\frac{B_D}{m}\frac{\alpha x}{\alpha x+\beta y}.
\end{equation}

Replacing eq.~\eqref{eq:leo2} and eq.~\eqref{eq:leo3}
in eq.~\eqref{eq:boh1},
\begin{equation}\label{eq:leo4}
x=-\frac{B_D}{m}\frac{\alpha x}{\alpha x+\beta y}+\frac{B_A+B_D}{m}\frac{\sqrt{\alpha x}}{\sqrt{\alpha x}+\sqrt{\beta y}}.
\end{equation}
We know that $\sum_{j=1}^n x_{A,j}=B_A$, thus
\begin{equation}\label{eq:gerard}
m(x+y)=B_A.
\end{equation}

From eq.~\eqref{eq:leo4} and eq.~\eqref{eq:gerard}, we obtain
\begin{eqnarray}
x={}&-\frac{B_D}{m}\frac{\alpha x}{(\alpha-\beta)x+\beta B_A/m}\\
&+\frac{B_A+B_D}{m}\frac{\sqrt{\alpha x}}{\sqrt{\alpha x}+\sqrt{\beta (B_A/m-x)}},
\end{eqnarray}
which corresponds to $x_A^*$ given in the Appendix.
\end{proof}

\section{Stackelberg leadership model}

For the case when there is an incumbent holding the market
and there is a challenger entering the market,
we consider the Stackelberg leadership model.
The Stackelberg leadership model
is a strategic game in which the leader
firm moves first and then the follower
firm moves afterwards.
To solve the Stackelberg model we need to find
the subgame perfect Nash equilibrium (SPNE) for each player
sequentially.
In our case, the defender is the leader who dominates the market and
the attacker is the follower who wants to enter the market.

From Proposition~\ref{prop:bamako}, the subgame perfect Nash equilibrium for the attacker
is given by eq.~\eqref{eq:boh1}.
Given this information, the leader
solves its own SPNE.

The Lagrangian for the incumbent is given by
\begin{equation}
\mathcal{L}({\bf x}_D)=\sum_{k=1}^n v_{D,k}\frac{x_{D,k}}{x_{A,k}+x_{D,k}}
-\mu\Big(\sum_{\ell=1}^n x_{D,\ell}-B_D\Big),
\end{equation}
where $\mu$ is the Lagrange multiplier.
Since the defender already knows the optimal allocation of resources for the challenger,
it incorporates this information into its Lagrangian,
\begin{align*}
\mathcal{L}({\bf x}_D)={}&
\sum_{k=1}^n \frac{v_{D,k}}{\sqrt{v_{A,k}x_{D,k}}}\frac{x_{D,k}}{B_A+B_D}\sum_{\ell=1}^n\sqrt{v_{A,\ell}x_{D,\ell}}\\
&-\mu\Big(\sum_{\ell=1}^n x_{D,\ell}-B_D\Big),
\end{align*}
or equivalently,
\begin{align*}
\mathcal{L}({\bf x}_D)={}&
\frac{1}{B_A+B_D}\sum_{\ell=1}^n\sqrt{v_{A,\ell}x_{D,\ell}}
\sum_{k=1}^n \frac{v_{D,k}\sqrt{x_{D,k}}}{\sqrt{v_{A,k}}}\\
&-\mu\Big(\sum_{\ell=1}^n x_{D,\ell}-B_D\Big),
\end{align*}

The necessary conditions for optimality are given by
the equations
\begin{align*}
\frac{\partial\mathcal{L}({\bf x}_D)}{\partial x_{D,k}}=
\frac{1}{B_A+B_D}
\Big(
\frac{\sqrt{v_{A,k}}}{2\sqrt{x_{D,k}}}
\Big)
\Big(
\sum_{\ell=1}^n \frac{v_{D,\ell}\sqrt{x_{D,\ell}}}{\sqrt{v_{A,\ell}}}
\Big)+\\
\frac{1}{B_A+B_D}
\Big(
\frac{v_{D,k}}{\sqrt{v_{A,k}}}
\frac{1}{2\sqrt{x_{D,k}}}\sum_{\ell=1}^n \sqrt{v_{A,\ell}x_{D,\ell}}
\Big)-\mu=0.
\end{align*}
Reordering terms,
\begin{align}\label{eq:toty}
&2\mu(B_A+B_D)=\\
&\frac{\sqrt{v_{A,k}}}{\sqrt{x_{D,k}}}\Big(
\sum_{\ell=1}^n \frac{v_{D,\ell}\sqrt{x_{D,\ell}}}{\sqrt{v_{A,\ell}}}
\Big)
+\frac{v_{D,k}}{\sqrt{v_{A,k}}}\frac{1}{\sqrt{x_{D,k}}}
\sum_{\ell=1}^n \sqrt{v_{A,\ell}x_{D,\ell}}.\nonumber
\end{align}
We will use this equation to compute the following cases.

\begin{proposition}
If $v_{i,j}=\alpha v_{-i,j}$ $\forall j\in\{1,\ldots,n\}$, $i\in\{A,D\}$,
with $\alpha>0$, then the Stackelberg equilibrium of the game is
\begin{equation}
x_{i,j}=B_i\frac{v_{i,j}}{V_i},
\end{equation}
where $V_i=\sum_{k=j}^n v_{i,j}$.
\end{proposition}

\begin{proof}
Following the proof of Proposition~\ref{prop:uganda},
when the valuations of the defender
are proportional to the valuations of the attacker,
the objective function for player $i$ is given by
$F_i({\bf x}_i,{\bf x}_{-i},{\bf v}_{-i})=\alpha F_i({\bf x}_i,{\bf x}_{-i},{\bf v}_{-i})$.
Therefore, player $i$ has an objective function equivalent
to $F_i({\bf x}_i,{\bf x}_{-i},{\bf v}_{-i})$.
In that case, the game is equivalent to a two-player zero-sum game 
and thus there is no difference between the Stackelberg equilibrium
and the Nash equilibrium given by Proposition~\ref{prop:uganda}.
\end{proof}

Contrary to the previous proposition, in the scenario 
considered in Proposition~\ref{prop:zaire1},
the strategies of the Stackelberg equilibrium
are very different from the strategies of the Nash equilibrium
given by  Proposition~\ref{prop:zaire1}
and different than the strategies previously considered.

\begin{proposition}\label{prop:tanzania1}
Assume that $n$ is even, so there exists \mbox{$m\in\mathbf{N}\setminus\{0\}$}
such that $n=2m$.
Assume that $v_{A,\ell}=v$ $\forall 1\le\ell\le n$,
$v_{D,k}=\alpha v$ $\forall 1\le k\le m$ and $v_{D,k'}=\beta v$ $\forall m+1\le k'\le n$,
for $\alpha,\beta>0$.
Then the Stackelberg equilibrium is given by
\begin{align*}
x_{D,k}&=\frac{B_D}{4m}\Big(2\pm\frac{\sqrt{2}(\beta-\alpha)}{\sqrt{\alpha^2 + \beta^2}}\Big)\quad\forall 1\le k\le m,\\
x_{D,k'}&=\frac{B_D}{4m}\Big(2\pm\frac{\sqrt{2}(\alpha-\beta)}{\sqrt{\alpha^2 + \beta^2}}\Big)\quad\forall m+1\le k'\le n.
\end{align*}
\end{proposition}

\begin{proof}
From eq.~\eqref{eq:toty},
for $1\le k\le m$,

\begin{align}
&2\mu(B_A+B_D)=
\frac{\sqrt{v}}{\sqrt{x_{D,k}}}
\alpha\sum_{\ell=1}^m\sqrt{vx_{D,\ell}}+\label{eq:kenya1}\\
&\frac{\sqrt{v}}{\sqrt{x_{D,k}}}\beta\sum_{\ell=m+1}^n\sqrt{vx_{D,\ell}}\nonumber
+
\frac{\sqrt{v}}{\sqrt{x_{D,k}}}
\left(
\alpha\sum_{\ell=1}^n\sqrt{vx_{D,\ell}}\right).
\end{align}

Equivalently,
\begin{align*}
&\sqrt{x_{D,k}}=\\
&\frac{v}{2\mu(B_A+B_D)}
\left(2\alpha\sum_{\ell=1}^m\sqrt{x_{D,\ell}}+(\alpha+\beta)\sum_{\ell=m+1}^n\sqrt{x_{D,\ell}}\right).
\end{align*}
We have the previous expression for every $k\in\{1,\ldots,m\}$, therefore
$x_{D,1}=x_{D,2}=\ldots=x_{D,m}\equiv x$.

Similarly, from eq.~\eqref{eq:toty},
for $m+1\le k'\le n$,
\begin{align}
2\mu(B_A+B_D)&=
\frac{\sqrt{v}}{\sqrt{x_{D,k'}}}
\left(
\alpha\sum\limits_{\ell=1}^m\sqrt{vx_{D,\ell}}+
\beta\sum\limits_{\ell=m+1}^n\sqrt{vx_{D,\ell}}\right)\nonumber
\\
&+
\frac{\sqrt{v}}{\sqrt{x_{D,k'}}}\left(
\beta\sum_{\ell=1}^n\sqrt{vx_{D,\ell}}\right)\label{eq:kenya2}
\end{align}
Equivalently,
\begin{align*}
&\sqrt{x_{D,k}}=\\
&\frac{v}{2\mu(B_A+B_D)}
\left(
(\alpha+\beta)\sum_{\ell=1}^m\sqrt{x_{D,\ell}}+
2\beta\sum_{\ell=m+1}^n\sqrt{x_{D,\ell}}\right).
\end{align*}
We have the previous expression for every $k\in\{m+1,\ldots,n\}$, therefore
$x_{D,m+1}=x_{D,m+2}=\ldots=x_{D,n}\equiv y$.

From the difference between eq.~\eqref{eq:kenya2} and eq.~\eqref{eq:kenya1},
\begin{equation}
\left(
\frac{\alpha+\beta}{\sqrt{y}}-\frac{2\alpha}{\sqrt{x}}
\right)m\sqrt{x}
=
\left(
\frac{\alpha+\beta}{\sqrt{x}}-\frac{2\beta}{\sqrt{y}}
\right)m\sqrt{y}.
\end{equation}
The solutions of the previous equation are given by
\begin{equation}
x=\frac{B_D}{4m}\Big(2\pm\frac{\sqrt{2}(\alpha-\beta)}{\sqrt{\alpha^2 + \beta^2}}\Big).
\end{equation}
Since $x+y=B_D/m$, then
\begin{eqnarray}
y=\frac{B_D}{4m}\Big(2\mp\frac{\sqrt{2}(\alpha-\beta)}{\sqrt{\alpha^2 + \beta^2}}\Big).
\end{eqnarray}

%
%
From eq.~\eqref{eq:boh1}, we obtain $x_{A,k}$ for both $1\le k\le m$ and $x_{A,k'}$ for $m+1\le k\le n$.
\end{proof}

\section{Simulations}

In this section, we compare through numerical simulations
the Nash equilibrium and the Stackelberg equilibrium
for the allocation game described above.

Consider that the number of potential customers \mbox{$n=100,000$} and
that within these potential customers we have two communities of equal size~\mbox{$m=50,000$}.
We consider that the budget allocated to capture the market
by the attacker is $B_A=200,000$, and we consider three different
scenarios for the budget of the defender: 
i) the defender has half the budget of the attacker~$B_D=B_A/2=100,000$,
ii) the defender and the attacker have the same budget~$B_D=B_A=200,000$, 
iii) the defender has two times the budget of the attacker~$B_D=2B_A=400,000$.

Assume that the network value of each customer for the challenger (attacker)
is the same $v_{A,\ell}=v=10\quad\forall 1\le\ell\le n$.
However, for the incumbent (challenger), each community has a different
network value $v_{D,\mathcal{C}_1}=v(1-\delta)=10(1-\delta)$ for the customers of the first community
and $v_{D,\mathcal{C}_2}=v(1+\delta)=10(1+\delta)$ for the customers of the second community,
where $\delta$ is a given parameter.

The Nash equilibrium (NE) was computed through Proposition~\ref{prop:zaire1} and
the Stackelberg equilibrium (SE) was computed through Proposition~\ref{prop:tanzania1}.
The percentage increase in profits was computed as $100*(SE-NE)/NE$.
The results are given in Figure~\ref{fig:rwanda2} and Figure~\ref{fig:rwanda1}.

\begin{figure}[h!]
  \centering
    \includegraphics[width=0.5\textwidth]{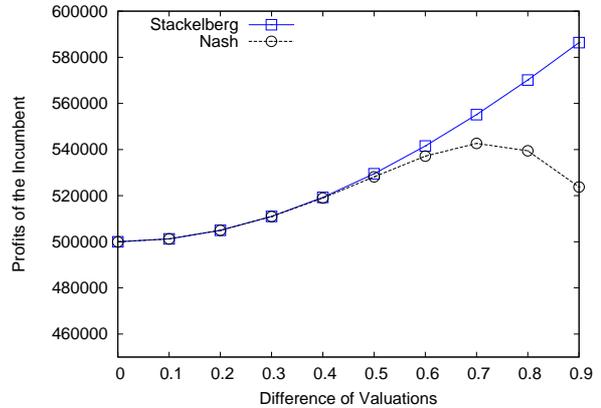}
  \caption{Profits for the incumbent (defender)
        vs the difference of valuations between communities ($\delta$) for equal budgets.}\label{fig:rwanda2}
\end{figure}

\begin{figure}[h!]
  \centering
    \includegraphics[width=0.5\textwidth]{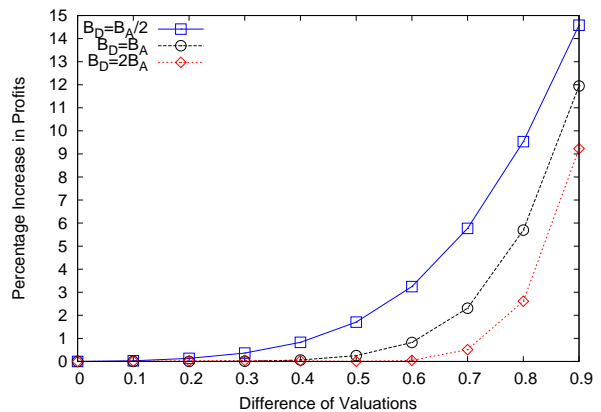}
  \caption{Percentage increase in profits for the incumbent (defender)
	vs the difference of valuations between communities ($\delta$)
	by committing to the Stackelberg leadership model versus Nash equilibrium.}\label{fig:rwanda1}
\end{figure}

In Figure~\ref{fig:rwanda2}, we observe that
for small difference of valuations
Stackelberg and Nash equilibria
give roughly the same profit.
However, when the difference of valuations
increases we have that while the profit
obtained by Stackelberg increases,
the profit of the Nash equilibrium after
a threshold decreases.

In Figure~\ref{fig:rwanda1}, we notice that for a small difference of valuations between communities,
both models give roughly the same profits, however when
the difference of valuations between communities grows,
the Stackelberg equilibrium gives much higher profits
than the Nash equilibrium.

Another interesting observations is that in
the case when the defender has smaller budget
compared with the attacker,
the difference in profits from both equilibria
is much higher compared with the scenario
when the defender has higher budget.

\begin{section}{Conclusions}
We have studied the case of two marketing campaigns
competing to maximize their profit from the
network value of the potential customers.
We have analyzed the following situations:
(a) when the decision of how many resources to
allocate to market to potential customers
is made simultaneously, and
(b) when the decision is sequential
and the incumbent foreseeing the arrival of the challenger
can commit to a strategy before its arrival.
\end{section}

\section*{Acknowledgments}
The work of A.~Silva was partially done in the context of the ADR ``Network Science" of the Joint Alcatel-Lucent Inria Lab.
The work of A.~Silva was partially carried out at LINCS (\url{www.lincs.fr}).

\bibliographystyle{hieeetr}
\bibliography{mybibfile}

\twocolumn[
  \begin{@twocolumnfalse}
\section*{Appendix}
The term $x_A^*$ was computed through Matlab Symbolic Toolbox from the equation:
\begin{align*}
x=-\frac{B_D}{m}\frac{\alpha x}{(\alpha-\beta)x+\beta B_A/m}
+\frac{B_A+B_D}{m}\frac{\sqrt{\alpha x}}{\sqrt{\alpha x}+\sqrt{\beta (B_A/m-x)}},
\end{align*}
and it is given by
\begin{align*}
&x_A^*=(((3^{(1/2)} ((B_D^2 \alpha^8 \beta^2 (4 B_A^4 \alpha^4 + 8 B_A^4 \alpha^2 \beta^2 + 4 B_A^4 \beta^4 + 4 B_A^3 B_D \alpha^4
- 8 B_A^3 B_D \alpha^3 \beta + 72 B_A^3 B_D \alpha^2 \beta^2 - 8 B_A^3 B_D \alpha \beta^3 \\
&+ 4 B_A^3 B_D \beta^4 - B_A^2 B_D^2 \alpha^4
- 24 B_A^2 B_D^2 \alpha^3 \beta + 146 B_A^2 B_D^2 \alpha^2 \beta^2 - 24 B_A^2 B_D^2 \alpha \beta^3 - B_A^2 B_D^2 \beta^4
- 16 B_A B_D^3 \alpha^3 \beta + 96 B_A B_D^3 \\& \alpha^2 \beta^2 - 16 B_A B_D^3 \alpha \beta^3 + 4 B_D^4 \alpha^3 \beta +
8 B_D^4 \alpha^2 \beta^2 + 4 B_D^4 \alpha \beta^3))/(m^6 (\alpha^2 - \beta^2)^4))^{(1/2)})/18 - (\alpha^3 (2 B_A \beta^2
- B_A \alpha^2 + 3 B_A \alpha \beta\\& 
+ 4 B_D \alpha \beta)^3)/(27 m^3 (\alpha + \beta)^3 (\alpha - \beta)^3) + (\alpha^3 \beta (2 B_A \beta^2
- B_A \alpha^2 + 3 B_A \alpha \beta + 4 B_D \alpha \beta) 
(3 B_A^2 \alpha \beta - 2 B_A^2 \alpha^2 + B_A^2 \beta^2\\& - 2 B_A B_D \alpha^2 +
6 B_A B_D \alpha \beta + B_D^2 \alpha^2 + B_D^2 \alpha \beta))/(6 m^3 (\alpha + \beta)^2 (\alpha - \beta)^3)
\\&  + (B_A \alpha^4 \beta^2 (B_A + B_D)^2)/(2 m^3 (\alpha + \beta) (\alpha - \beta)^2))^{(1/3)}
- (\alpha (2 B_A \beta^2 - B_A \alpha^2 +
3 B_A \alpha \beta + 4 B_D \alpha \beta))/(3 m (\alpha^2 - \beta^2))
\\&+(\alpha^2 (B_A^2 \alpha^4 + 2 B_A^2 \alpha^2 \beta^2 + B_A^2 \beta^4
- 2 B_A B_D \alpha^3 \beta + 12 B_A B_D \alpha^2 \beta^2 - 2 B_A B_D \alpha \beta^3 - 3 B_D^2 \alpha^3 \beta + 10 B_D^2 \alpha^2 \beta^2\\
&- 3 B_D^2 \alpha \beta^3))/(9 m^2 (\alpha^2 - \beta^2)^2 ((3^{(1/2)} ((B_D^2 \alpha^8 \beta^2 (4 B_A^4 \alpha^4
+ 4 B_A^4 \beta^4 + 4 B_A^3 B_D \alpha^4 + 4 B_A^3 B_D \beta^4 + 4 B_D^4 \alpha \beta^3 + 4 B_D^4 \alpha^3 \beta\\&
- B_A^2 B_D^2 \alpha^4
- B_A^2 B_D^2 \beta^4 + 8 B_A^4 \alpha^2 \beta^2 + 8 B_D^4 \alpha^2 \beta^2
+ 146 B_A^2 B_D^2 \alpha^2 \beta^2 - 16 B_A B_D^3 \alpha \beta^3
- 16 B_A B_D^3 \alpha^3 \beta - 8 B_A^3 B_D \alpha \beta^3
\\& - 8 B_A^3 B_D \alpha^3 \beta + 96 B_A B_D^3 \alpha^2 \beta^2 - 24 B_A^2 B_D^2 \alpha \beta^3
- 24 B_A^2 B_D^2 \alpha^3 \beta + 72 B_A^3 B_D \alpha^2 \beta^2))/(m^6 (\alpha^2 - \beta^2)^4))^{(1/2)})/18 
\\&- (\alpha^3 (2 B_A \beta^2
- B_A \alpha^2 + 3 B_A \alpha \beta + 4 B_D \alpha \beta)^3)/(27 m^3 (\alpha + \beta)^3 (\alpha - \beta)^3) 
\\& + (\alpha^3 \beta (2 B_A \beta^2
- B_A \alpha^2 + 3 B_A \alpha \beta + 4 B_D \alpha \beta) (B_A^2 \beta^2 - 2 B_A^2 \alpha^2 + B_D^2 \alpha^2 - 2 B_A B_D \alpha^2 +\\
& 3 B_A^2 \alpha \beta + B_D^2 \alpha \beta + 6 B_A B_D \alpha \beta))/(6 m^3 (\alpha + \beta)^2 (\alpha - \beta)^3) +
(B_A \alpha^4 \beta^2 (B_A + B_D)^2)/(2 m^3 (\alpha + \beta) (\alpha - \beta)^2))^{(1/3)}))/\alpha.\\
\end{align*}
\end{@twocolumnfalse}
]

\end{document}